\documentclass[12pt,a4paper]{article}
\nonstopmode
\usepackage{epsfig}
\usepackage[a4paper]{geometry}
\usepackage{float}
\usepackage[stable]{footmisc}
\usepackage[utf8]{inputenc}
\usepackage{a4wide}
\usepackage{amsmath,amstext,amsthm,amssymb}
\usepackage{amsfonts}
\usepackage{framed}
\usepackage{graphicx}
\usepackage{color}
\usepackage[most]{tcolorbox}
\usepackage{bbold}
\usepackage{bbm}
\usepackage[normalem]{ulem}
\usepackage{rotating} 
\usepackage{slashed} 
\usepackage{cancel} 
\usepackage{braket} 
\usepackage{amstext}
\usepackage{array}
\usepackage{setspace}
\usepackage{hyperref}
\usepackage{overpic}
\usepackage{bbm}
\usepackage{cite}
\usepackage{ulem}
\usepackage{lipsum}

\usepackage{tikz}
\usetikzlibrary{calc}
\usetikzlibrary{shapes.misc}
\tikzstyle{cross}=[
    cross out,
    minimum size=2*(#1),
    inner sep=0pt,
    outer sep=0pt
]

\usepackage{cleveref}
\usepackage{algorithm}
\usepackage{algpseudocode}
\usepackage{diagbox}
\usepackage{blkarray}
\usepackage{placeins}
\usepackage{tablefootnote}

\newtheorem{lemma}{Lemma}
\newtheorem{proposition}{Proposition}

\newcommand{\DeclareRuneSeparators}[1]{} 

\input{arm.fd}

\let\baraccent=\=
\renewcommand{\=}[1]{\stackrel{#1}{=}}

\DeclareSymbolFontAlphabet{\mathbb}{AMSb}

\begin{document}
	
	\pagestyle{plain}

	\makeatletter
	\@addtoreset{equation}{section}
	\makeatother
	\renewcommand{\theequation}{\thesection.\arabic{equation}}
	\pagestyle{empty}
	
	\vspace{0.5cm}
	
	\begin{center}
		
		{\LARGE \bf{Efficient Algorithm for Generating Homotopy Inequivalent Calabi-Yaus}
			\\[10mm]}
	\end{center}

	\begin{center}
		\scalebox{0.95}[0.95]{{\fontsize{14}{30}\selectfont Nate MacFadden$^{a}$}}
	\end{center}

	\begin{center}
		\vspace{0.25 cm}
		\textsl{$^{a}$Department of Physics, Cornell University, Ithaca, NY 14853 USA}\\
		
		\vspace{1cm}
		\normalsize{\bf Abstract} \\[8mm]

	\end{center}
        We present an algorithm for efficiently exploring inequivalent Calabi-Yau threefold hypersurfaces in toric varieties.
        A direct enumeration of fine, regular, star triangulations (FRSTs) of polytopes in the Kreuzer-Skarke database is foreseeably impossible due to the large count of distinct FRSTs. Moreover, such an enumeration is needlessly redundant because many such triangulations have the same restrictions to 2-faces and hence, by Wall’s theorem, lead to equivalent Calabi-Yau threefolds.
        We show that this redundancy can be circumvented by finding a height vector in the strict interior of the intersection of the secondary cones associated with each 2-face triangulation.
        We demonstrate that such triangulations are generated with orders of magnitude fewer operations than the na\"ive approach of generating all FRSTs and selecting only those differing on 2-faces.
        Similar methods are also presented to directly generate (the support of) the secondary subfan of all fine triangulations, relevant for random sampling of FRSTs.
	\begin{center}
		\begin{minipage}[h]{15.0cm}

		\end{minipage}
	\end{center}
	\newpage
	\setcounter{page}{1}
	\pagestyle{plain}
	\renewcommand{\thefootnote}{\arabic{footnote}}
	\setcounter{footnote}{0}
	%
	%
	\tableofcontents
	\newpage

\section{Introduction}

In order to understand quantum gravity, it is important to study compactifications of string theory. One way to do this is by studying Calabi-Yau threefold hypersurfaces in toric varieties (CYs). The natural object of interest, then, is the Kreuzer-Skarke database (KS)\cite{Kreuzer:2000xy}, a complete collection of all $473,800,776$ $4$D reflexive polytopes. This is because KS defines the largest known number of CYs: any fine, regular, star triangulation (FRST) of any polytope in KS specifies the topological data defining a CY. Briefly, these adjectives mean
\begin{enumerate}
    \item `fine' $\implies$ every lattice point in the polytope is a vertex of at least one simplex in the triangulation,
    \item `regular' $\implies$ the triangulation can be constructed following the procedure laid out in \cref{subsec:regularity}, and
    \item `star' $\implies$ the origin is a vertex of every simplex in the triangulation.
\end{enumerate}
The reasons for imposing such restrictions are discussed in detail in \cite{Demirtas:2020dbm}.

Brute-force approaches to enumerate this population of CYs (i.e., generate all FRSTs, map to their associated CYs) are limited primarily by the count of FRSTs: $N_\mathrm{FRST}<1.53\times10^{928}$\cite{Demirtas:2020dbm}. Fortunately, as is shown in \cite{Demirtas:2020dbm}, this collection of CYs is very redundant: topologically equivalent CYs define physically equivalent solutions of string theory and, by Wall's theorem\cite{Wall:1966}, any two FRSTs, $\mathcal{T}_1$ and $\mathcal{T}_2$, of a polytope $\Delta\in\mathrm{KS}$ with the same $2$-face restrictions generate topologically equivalent CYs. Thus, it is sufficient to study only FRSTs with unique $2$-face restrictions, hereon denoted `NTFE' for `non-$2$-face-equivalent'. There are ``only" $N_\mathrm{NTFE}<1.65\times10^{428}$\cite{Demirtas:2020dbm}\footnote{With some work, one can tighten the bound on $N_\textrm{NTFE}$ \textit{significantly}, from both sides\cite{MacFadden:WIP}. Exact bounds are not quoted here (no spoilers!) but, for a teaser, the upper bound drops by $>50$ orders of magnitude.} NTFE FRSTs. The set of NTFE FRSTs still contains redundancies due to non-trivial basis transformations mapping CYs into one-another, thus providing hope for even smaller counts of inequivalent CYs, but those redundancies will be harder to deal with - we do not concern them in this work.

There are two primary strategies to utilize this $500$ order-of-magnitude redundancy:
\begin{enumerate}
    \item (the `mod approach') generate all FRSTs and then mod out by $2$-face equivalence before running any of the expensive CY calculations or
    \item (the `on-demand approach') \textit{somehow} directly generate NTFE FRSTs.
\end{enumerate}
The mod approach, used in numerous studies such as \cite{Demirtas:2021gsq,Gendler:2022ztv,Gendler:2023hwg}, is nice because it is both simple and enables order-of-magnitude speedups over the aforementioned `brute-force' algorithm. However, as we discuss in \cref{sec:naive}, the mod approach is ultimately limited by the large count of FRSTs. For example, the largest $h^{1,1}$-value for which one can generate all FRSTs of a typical polytope (in reasonable time on a modern personal computer) is $h^{1,1}\lesssim10$. This is true even if the count of NTFE FRSTs of said polytope is relatively small, such that the subsequent CY-focused calculations are relatively quick.

The difficulty in generating all FRSTs motivates the on-demand approach: if one could directly generate NTFE FRSTs, the (highly redundant) count of all FRSTs would become irrelevant. Before this paper, as far as the author knows, there was no known method for such direct enumeration. However, as we show in \cref{sec:on-demand}, such direct generation can be achieved both simply and efficiently. In brief, one can (efficiently) either directly generate an FRST with user-specified $2$-face restrictions or prove that no such FRST exists. This on-demand generation algorithm is the main result of this paper.

Finally, in \cref{sec:fan}, we provide an adaptation of the on-demand NTFE algorithm. This adaptation, instead of generating FRSTs, generates an object called the `support of the secondary subfan of fine triangulations'. Roughly, this is an object describing \textit{all} NTFE FRSTs of a given polytope, even when direct enumeration is infeasible (e.g., for polytopes with $h^{1,1}\gtrsim 100$). The provided algorithm both provably generates the desired cone and it does so very quickly: the run time for even the largest polytope ($h^{1,1}=491$) is only $<1$min on modern personal computers. Run time is even quicker for smaller $\Delta\in\mathrm{KS}$. (A variant of) this subfan has already been used successfully for fairly sampling FRSTs\cite{Demirtas:2020dbm}, and more algorithms are currently being developed to further utilize it.

\section{The `mod' approach}
\label{sec:naive}

As we have seen, the population of physical interest when studying CYs from KS is that of NTFE FRSTs. The most direct method of enumerating only the NTFE FRSTs of a given polytope $\Delta\in\mathrm{KS}$ is by generating all FRSTs and then modding out by $2$-face equivalence. A nice way to visualize this procedure is through the use of the (bistellar) flip graph of $\Delta$. This is the graph in which nodes represent FRSTs and edges represent bistellar flips between the FRSTs. Such graphs can na\"ively be viewed as a portion of the string landscape - each node representing a solution of string theory.

In terms of flip graphs, the modding-out procedure is simply contracting nodes if they have the same $2$-face restrictions\footnote{It is not hard to see that the subgraph of FRSTs that are $2$-face equivalent to some triangulation, $\mathcal{T}$, is connected. Thus, contraction is indeed the same as modding-out.}. For example, consider the flip graph of the $0$th polytope in KS (lattice N) with $h^{1,1}=5$\footnote{The lattice points defining this polytope are $(-2,2,1,-1)$, $(-1,1,1,0)$, $(0,-1,-1,0)$, $(0,0,0,0)$, $(0,0,0,1)$, $(0,0,1,0)$, $(0,1,0,0)$, $(1,-2,1,1)$, $(1,0,0,0)$, and $(1,1,-1,-1)$.}, plotted on the left side of \cref{fig:flip_graphs}. This graph has $N=142$ nodes but, after modding out by $2$-face equivalence, only $N=2$ of them survive (see right side of \cref{fig:flip_graphs}). Since
\begin{enumerate}
    \item all $2$-face equivalent nodes represent physically-identical solutions of string theory,
    \item the NTFE flip graphs (i.e., the flip graphs after contraction) are typically much smaller than the FRST flip graphs (as we will soon argue more generally), and
    \item it is relatively cheap to mod out by $2$-face equivalence,
\end{enumerate}
it is clearly preferable to always perform this modding out before carrying out any expensive computations on the CYs. In other words, while the left side of \cref{fig:flip_graphs} is the na\"ive, initial estimate of some region of the landscape, it actually contains many copies of the same physical solutions! When one takes Wall's theorem into consideration, the landscape collapses to something much smaller.

\begin{figure}[t]
    \centering
    \hspace{-1cm}
    \begin{minipage}{0.5\textwidth}
        \centering
        \includegraphics[width=0.8\linewidth]{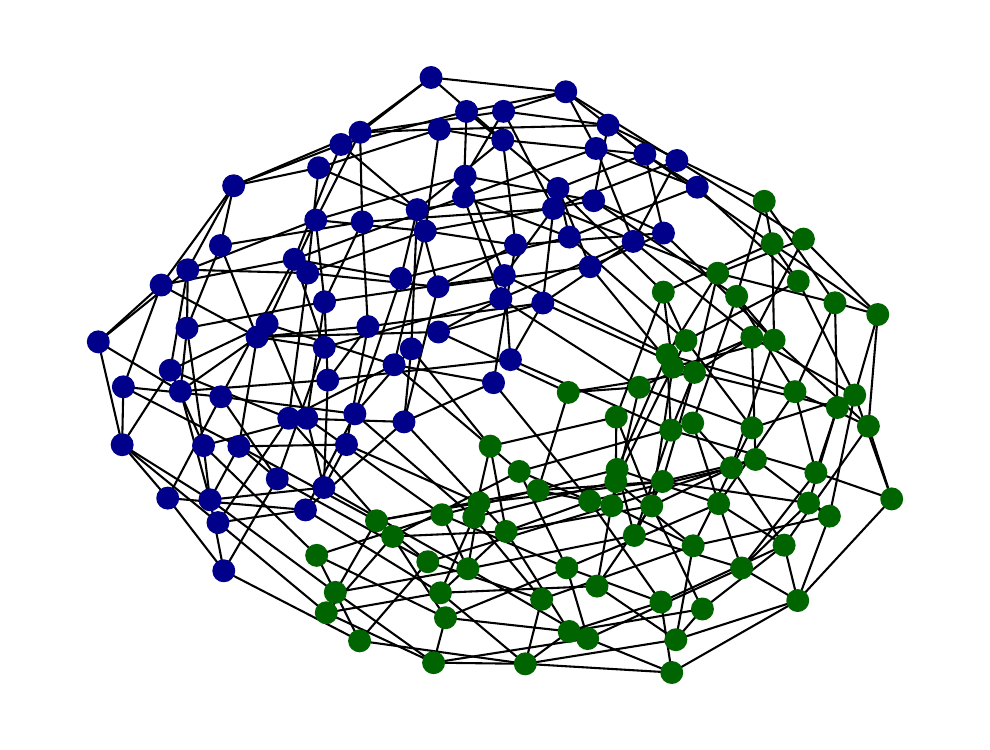}
    \end{minipage}
    \begin{minipage}{0.5\textwidth}
        \centering
        \includegraphics[width=0.7\linewidth]{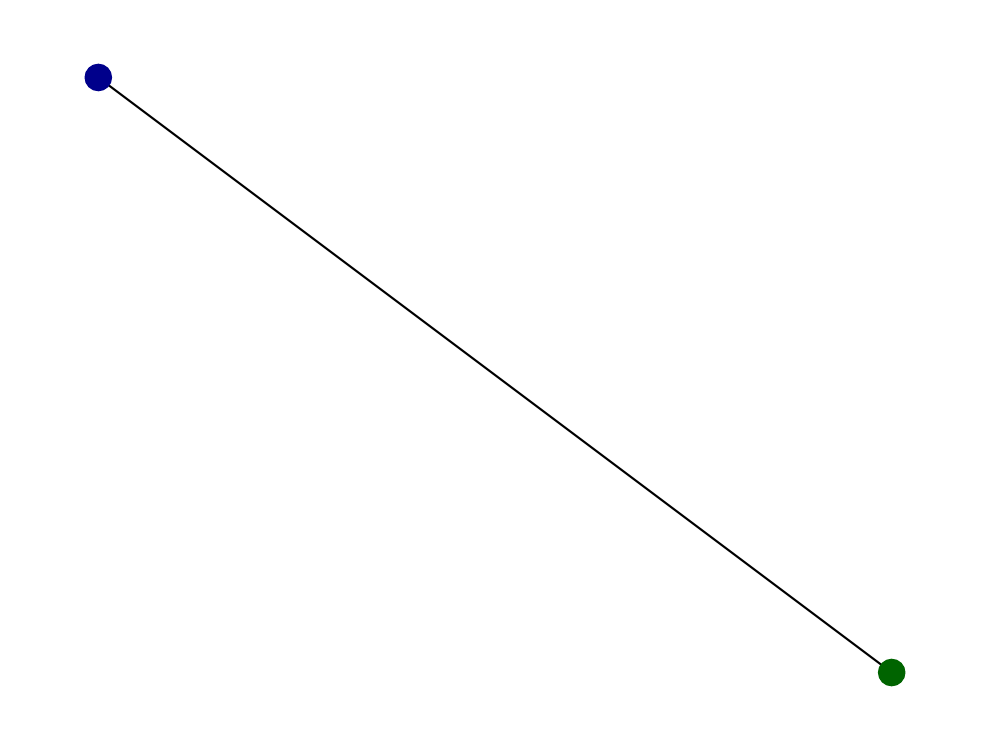}
    \end{minipage}
    
    \begin{tikzpicture}[remember picture, overlay]
        \draw[->, line width=0.8mm] (-1,3) .. controls (0,3.2) .. (1.5,3);
    \end{tikzpicture}
    
    \caption{\emph{Left}: the flip graph of the $0$th polytope in KS (lattice N) with $h^{1,1}=5$, colored by $2$-face equivalence. \emph{Right}: after contracting by $2$-face equivalence, only $2$ nodes out of the original $142$ survive.}
    \label{fig:flip_graphs}
\end{figure}
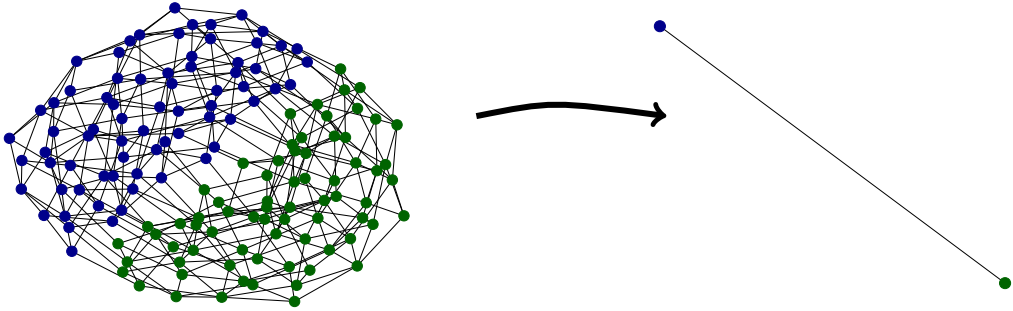

While the reduction in \cref{fig:flip_graphs} is sizeable, it is actually very modest compared to the gains achievable in KS. This is primarily due to the small value of $h^{1,1}$: even marginal increases in $h^{1,1}$ cause significant increases in the reduction from all FRSTs to only the NTFE FRSTs. For example, the $0$th polytope with $h^{1,1}=8$ contains only $N=4$ NTFE FRSTs despite having $N=1171$ FRSTs. As $h^{1,1}$ further increases, this redundancy of FRSTs (when the real population of interest are the NTFE FRSTs) grows rapidly. Even at relatively small values of $h^{1,1}$ (such as $10$), the enumeration of all FRSTs can be prohibitive in time and memory. This growing redundancy will be demonstrated in detail in \cref{subsec:benchmarks}.

This scaling with $h^{1,1}$ could have been anticipated by the bounds\cite{Demirtas:2020dbm} on the count of (NTFE) FRSTs for any single polytope in KS:
\begin{align}
    \textrm{\# FRSTs/poly} &\lesssim 10^{-5.31}\,10^{1.91\,h^{1,1}} \textrm{ and }\\
    \textrm{\# NTFE FRSTs/poly} &\lesssim 10^{-15.45}\,10^{0.90\,h^{1,1}},
\end{align}
plotted in \cref{fig:redundancy}. If the bound on $\textrm{\# FRSTs/poly}$ is not exponentially loose, then there is an exponentially-increasing redundancy of FRSTs over NTFE FRSTs. This redundancy ultimately limits the mod approach: quickly with $h^{1,1}$, the cost of generating all FRSTs (only to throw out most of them) limits the scope of the study.

\begin{figure}[t]
    \centering
    \includegraphics[width=0.75\textwidth]{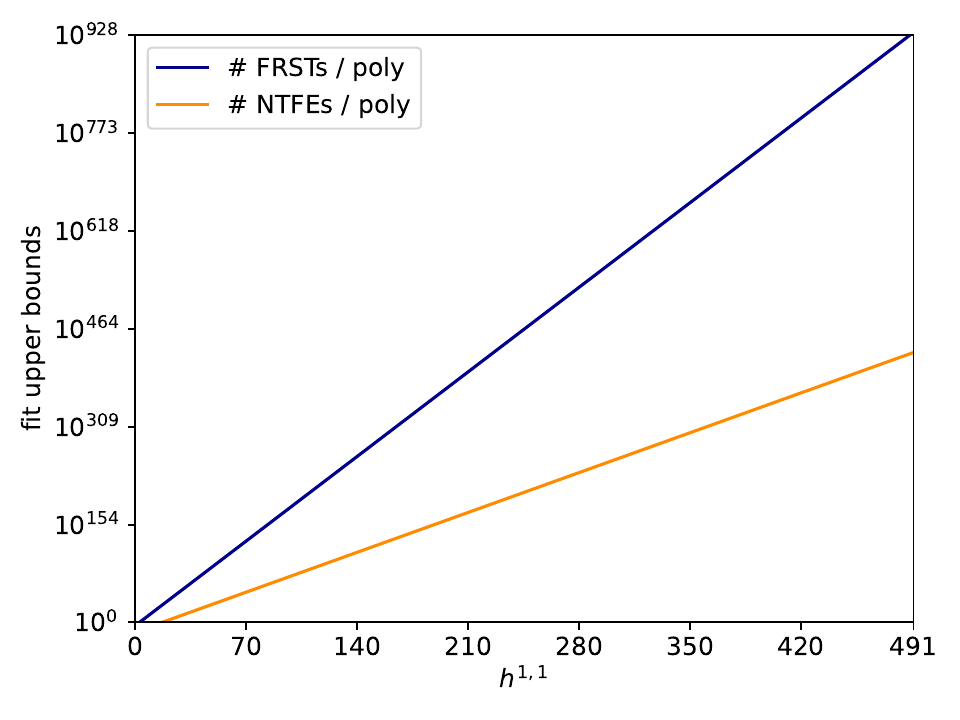}
    \caption{(Data/fits from \cite{Demirtas:2020dbm}). Approximate bounds on the number of FRSTs (blue; $N\lesssim 10^{-5.31}\,10^{1.91\,h^{1,1}}$) and the number of NTFEs (yellow; $N\lesssim 10^{-15.45}\,10^{0.90\,h^{1,1}}$) per $\Delta\in\mathrm{KS}$, as a function of $h^{1,1}$. While both bounds scale exponentially with $h^{1,1}$, the bound on FRSTs grows exponentially quicker than that on NTFEs.}
    \label{fig:redundancy}
\end{figure}

Clearly, it would be greatly beneficial to directly generate the NTFE FRSTs, rather than stepping through all FRSTs. This would maintain the speedup of having fewer costly operations on CYs, but it would achieve the additional speedup in reduced\footnote{`Reduced' since there still may be non-trivial basis transformations mapping CYs into one-another. For small $h^{1,1}$, these equivalences are studied in detail in \cite{Gendler:WIP}.} generation of topologically equivalent CYs in intermediate steps. In terms of flip graphs, the goal is to directly generate the nodes of the graph after contraction (right side of \cref{fig:flip_graphs}). In the following section, we demonstrate how to do this.

\section{On-demand generation}
\label{sec:on-demand}

In this section, we demonstrate how to generate NTFE FRSTs `on-demand'. This discussion leans on regularity, so we must first recall some definitions and notation.

\subsection{Regularity}
\label{subsec:regularity}
Let $\mathcal{A} = \{p_1, \dots, p_n\}$ be a set of points $p_i\in\mathbb{R}^d$ and let ${p_i}^j$ be the $j$th coordinate of $p_i$. A triangulation of $\mathcal{A}$ is called `regular' if it is obtainable via\footnote{N.B.: Other definitions of regularity exist, but the one provided is most useful for current discussions.}:
\begin{enumerate}
    \item `lifting' each point $p_i$ by some height $h_i$: $p_i\to\tilde{p}_i = ({p_i}^1, \dots, {p_i}^d, h_i)$ and
    \item projecting out the final coordinate of the `lower faces' (those whose outwards-facing normal vectors have a negative $(d+1)$\textsuperscript{st} component) of $C = \mathrm{conv}\left(\{\tilde{p}_0, \dots, \tilde{p}_n\}\right)$.
\end{enumerate}
For an illustration of this lifting/projecting procedure for generating regular triangulations, see \cref{fig:lifting_basic}.

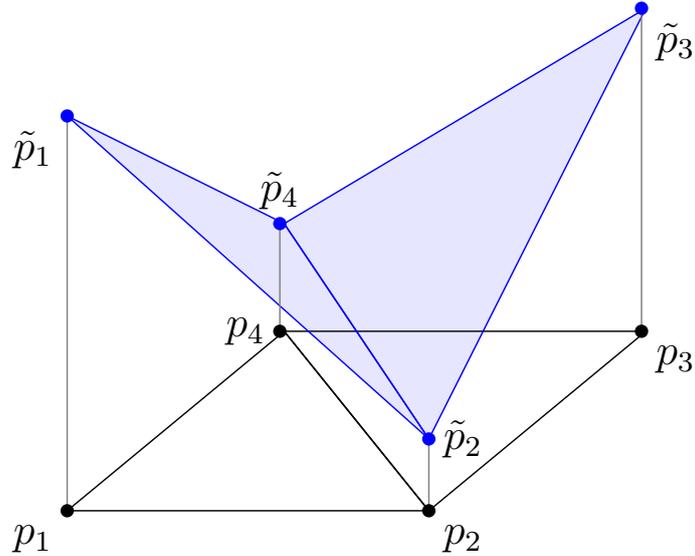
\begin{figure}[t]
    \centering
    \resizebox{0.6\textwidth}{!}{
    \begin{tikzpicture}[scale=3.5]
        \draw[black] (0, 0) -- (1, 0) -- (0+0.6, 0.5) -- cycle;
        \draw[black] (1+0.6, 0.5) -- (1, 0) -- (0+0.6, 0.5) -- cycle;
    
        \fill[blue,opacity=0.1] (0, 0+1.1) -- (1, 0+0.2) -- (0+0.6, 0.5+0.3) -- cycle;
        \fill[blue,opacity=0.1] (1+0.6, 0.5+0.9) -- (1, 0+0.2) -- (0+0.6, 0.5+0.3) -- cycle;
        \draw[blue] (0, 0+1.1) -- (1, 0+0.2) -- (0+0.6, 0.5+0.3) -- cycle;
        \draw[blue] (1+0.6, 0.5+0.9) -- (1, 0+0.2) -- (0+0.6, 0.5+0.3) -- cycle;
    
        \foreach \x in {0,1} {
            \foreach \y in {0,0.5} {
                \pgfmathsetmacro\height{
                    ifthenelse(\x==0 && \y==0, 1.1,
                    ifthenelse(\x==1 && \y==0, 0.2,
                    ifthenelse(\x==1 && \y==0.5, 0.9, 
                    ifthenelse(\x==0 && \y==0.5, 0.3, ))))
                }
                \draw[thin, gray] (\x+1.176*\y,\y) -- (\x+1.176*\y,\y+\height);
            }
        }
    
        \node[fill=black, circle, inner sep=1.3pt] at (0,0) {};
        \node[fill=black, circle, inner sep=1.3pt] at (1,0) {};
        \node[fill=black, circle, inner sep=1.3pt] at (1+1.176*0.5,0.5) {};
        \node[fill=black, circle, inner sep=1.3pt] at (0+1.176*0.5,0.5) {};

        \node[below left] at (0,0) {$p_1$};
        \node[below right] at (1,0) {$p_2$};
        \node[below right] at (1+1.176*0.5,0.5) {$p_3$};
        \node[left] at (0+1.176*0.5,0.5) {$p_4$};
    
        \node[fill=blue, circle, inner sep=1.3pt] at (0,0+1.1) {};
        \node[fill=blue, circle, inner sep=1.3pt] at (1,0+0.2) {};
        \node[fill=blue, circle, inner sep=1.3pt] at (1+1.176*0.5,0.5+0.9) {};
        \node[fill=blue, circle, inner sep=1.3pt] at (0+1.176*0.5,0.5+0.3) {};

        \node[below left] at (0,0+1.1) {$\tilde{p}_1$};
        \node[right] at (1,0+0.2) {$\tilde{p}_2$};
        \node[below right] at (1+1.176*0.5,0.5+0.9) {$\tilde{p}_3$};
        \node[above] at (0+1.176*0.5,0.5+0.3) {$\tilde{p}_4$};

    \end{tikzpicture}
    }
    \caption{The lifting of the point set $\mathcal{A}=\{p_1,p_2,p_3,p_4\}$ by heights $h_1=1.1$, $h_2=0.2$, $h_3=0.9$, and $h_4=0.3$. The convex hull of the lifted point set, $\tilde{\mathcal{A}}$, is a $3$-simplex whose lower two faces are plotted in blue. Projecting out the lifted coordinate generates the regular triangulation plotted in black.}
    \label{fig:lifting_basic}
\end{figure}

\subsubsection{The secondary cone}
\label{subsubsec:secondary_cone}

The heights generating a given regular triangulation are not unique. For example, consider the point configuration $\mathcal{A}=\{p_1,p_2,p_3,p_4\}$ as in \cref{fig:lifting_basic}. These points are lifted by heights $h_1=1.1$, $h_2=0.2$, $h_3=0.9$, and $h_4=0.3$, resulting in the convex hull $C=\mathrm{conv}(\{\tilde{p}_1, \tilde{p}_2, \tilde{p}_3, \tilde{p}_4\})$, a $3$-simplex. The lower faces of $C$ are then outlined in blue, corresponding to the regular triangulation plotted in black. Many different heights lead to the same triangulation: for example, one is free to increase $h_1$ and/or $h_3$ (or, conversely, decrease $h_2$ and/or $h_4$), without changing (the projection of) the lower faces of $C$.

To understand the space of heights that are allowed for a given triangulation, $\mathcal{T}$, organize the heights in a vector\footnote{Since vectors are ordered, a choice of point ordering must be made. We will (arbitrarily) always choose the order in which the point configuration is written on the page.}, $h = (h_1,\, \dots,\, h_n)$, called the `height vector' living in `height space'. E.g., the height vector associated to \cref{fig:lifting_basic} is $h=(1.1,\, 0.2,\, 0.9,\, 0.3)$. It is not difficult to show\footnote{Here is a brief argument: note that, if $h$ generates a triangulation $\mathcal{T}$, then so does $c\, h$ for any $c>0$. Likewise, if $g$ also generates $\mathcal{T}$, then it is not hard to see that $h+g$ generates $\mathcal{T}$.} that the collection of all height vectors generating the same triangulation, $\mathcal{T}$, forms the interior of a polyhedral cone called the `secondary cone'\cite{Loera:2010}. That is, $h$, generates the triangulation, $\mathcal{T}$, if and only if $h$ is in the strict interior of $\mathcal{T}$'s secondary cone.

The secondary cone is most directly represented in H-representation for which, in $2$D, the inwards-facing defining hyperplane inequalities, $H$, may be calculated as in \cref{alg:cpl_ineqs}\cite{Berglund:1995}\footnote{In fact, as seen in \cref{app:gerald_verification}, the explicit null-space calculations used in this algorithm are not necessary. The null-spaces can be identified by simply counting points.}. For example, following this algorithm, one finds a single hyperplane inequality associated to \cref{fig:lifting_basic}: $h_2 + h_4 \leq h_1 + h_3$. Any heights strictly obeying this inequality generate the displayed triangulation. This matches intuition: all that matters in this case is that every interior point of the line $(\tilde{p}_2, \tilde{p}_4)$ is below the corresponding point in either $\mathrm{conv}({\{\tilde{p}_1, \tilde{p}_2, \tilde{p}_3\}})$ or in $\mathrm{conv}({\{\tilde{p}_1, \tilde{p}_3, \tilde{p}_4\}})$.

\begin{algorithm}[t]
\caption{secondary cone}
\label{alg:cpl_ineqs}
\begin{algorithmic}
\State let \textit{H} be an empty array
\For{adjacent simplices $\{p_{n1}, p_{s1}, p_{s2}\}$ and $\{p_{n2}, p_{s1}, p_{s2}\}$}
    \State let $I=[s1,s2,n1,n2]$
    \For{column $c\in \mathrm{null}\left( \begin{bmatrix}
    p_{I[0]} & p_{I[1]} & p_{I[2]} & p_{I[3]} \\
    1 & 1 & 1 & 1
    \end{bmatrix} \right)$}
        \State let $n$ be a vector with $n^i=\begin{cases}
            c^j & i=I[j]\\
            0 & \textrm{otherwise}
        \end{cases}$
        \If{$n^{I[2]}<0$}
            \State append $(-n)^\mathrm{T}$ to $H$
        \Else
            \State append $n^\mathrm{T}$ to $H$
        \EndIf
    \EndFor
\EndFor
\State \textbf{return} $H$
\end{algorithmic}
\end{algorithm}

\subsubsection{Multiple secondary cones}
\label{subsubsec:secondary_cones}

Consider the case of multiple point configurations, $\mathcal{A}_1, \dots, \mathcal{A}_n$, with corresponding regular triangulations $\mathcal{T}_1, \dots, \mathcal{T}_n$. As we have discussed, due to regularity, each of these triangulations $\mathcal{T}_i$ can be generated by a height vector, $h_i$. Stronger, a height vector, $h_i$, generates $\mathcal{T}_i$ if and only if $H_i\, h_i>0$, where the rows of $H_i$ are the inwards-facing hyperplane normals defining $\mathcal{T}_i$'s secondary cone.

Consider the union $\mathcal{A} = \bigcup_{i=1}^n \mathcal{A}_i$. We want a height vector, $h$, which `simultaneously generates' all triangulations $\mathcal{T}_1, \dots, \mathcal{T}_n$. By that, we mean: let $\Pi_i$ be a projection matrix from the height space associated to $\mathcal{A}$ to that associated with $\mathcal{A}_i$. Then, we seek an $h$ such that $H_i\, \left(\Pi_i\, h\right) > 0$ for all $i$. Define
\begin{equation}
    H = \begin{bmatrix} H_1 \Pi_1\\ \vdots\\ H_n \Pi_n \end{bmatrix}.
\end{equation}
This object makes sense because each $H_i \Pi_i$ object has width $\left| \mathcal{A} \right|$, and thus can be vertically stacked. Thus, the goal\footnote{Note: it is OK to embed in an even higher dimensional space. By doing so, one is simply introducing more points in $\mathcal{A}$. The heights of these `extra' points will simply not be constrained by $H$.} is an $h$ such that $H\, h>0$. This is trivially done with linear programming.

\subsection{Applied to NTFE FRSTs}
\label{subsec:on-demand}

The application of \cref{subsubsec:secondary_cones} to the study of NTFE FRSTs will be direct, but we first must make some observations:
\begin{enumerate}
    \item Fineness with respect to all points is not necessary for the CY. Instead, all that is necessary is fineness with respect to points appearing on $2$-faces of the polytope\cite{Braun:2017}.
    \item Any non-star, regular triangulation can be converted into a star triangulation without affecting the $2$-face triangulations. This is simply by lowering the height of the origin until it appears in all simplices. Thus, the `star' requirement is semi-ignorable.
\end{enumerate}
Thus, all one needs to do is generate fine, regular triangulations, where `fine' will hereon mean `fine with respect to $2$-faces'. We will denote such triangulations as `FR(S)Ts', to emphasize that they are not truly star, but that there is a direct procedure to convert them to a star triangulation without affecting the $2$-faces. Those triangulations that are genuinely star will still be denoted as FRSTs.

With these notes, the procedure in \cref{subsubsec:secondary_cones} is exactly what we need: if one lets
\begin{enumerate}
    \item $\mathcal{A}$ be the point set of the polytope of interest,
    \item $\mathcal{A}_i$ be the point set of each $2$-face (in an arbitrary $2$D basis),
    \item $\Pi_i$ be the projection from the ordering of $\mathcal{A}$ to $\mathcal{A}_i$, and
    \item $H_i$ be the (inwards facing) hyperplane inequalities of each $2$-face's secondary cone,
\end{enumerate}
then direct application of \cref{subsubsec:secondary_cones} generates NTFE FR(S)Ts on-demand. This is the main result of this paper.

This calculation has a simple geometric interpretation. The objects $H_i\,\Pi_i$ each represent the embedding $\mathcal{T}_i$'s secondary cone in the height space of $\mathcal{A}$. By vertically stacking such $H_i\,\Pi_i$, we are taking the intersection of such cones. Thus, $H$ is just the intersection of all $2$-face secondary cones. This is effectively equivalent to the problem of finding a height vector associated to an FRST, except we have dropped all hyperplane inequalities constraining points not on $2$-faces.

It is important to stress: this algorithm, the main result, follows from three crucial observations:
\begin{enumerate}
    \item CYs only require fineness on $2$-faces,
    \item the height of the origin is ignorable in intermediate steps, and
    \item by Wall's theorem, only FRSTs with distinct $2$-faces are physically relevant.
\end{enumerate}
Together, these observations suggest that, \textbf{when constructing toric solutions to string theory from KS, one loses nothing by restricting attention to $2$-faces}. The other points \textit{do} contain information (especially the origin), but those points are not needed for specifying NTFEs and said information is recoverable. Needless to say, $2$D polytopes/triangulations are much easier to handle than $4$D ones.

\subsection{Benchmarks}
\label{subsec:benchmarks}

The entire purpose of the on-demand algorithm is its efficiency, so we provide some benchmarks against the mod algorithm. Note that the mod algorithm comprises two steps: generating all FRSTs and then modding out by them. Thus, it suffices to show that the on-demand algorithm is more efficient than \textit{just} the step of generating all FRSTs.

To directly generate NTFE FRSTs, we implement the on-demand algorithm in Python, using the CYTools \cite{Demirtas2022cytools} framework. To generate all FRSTs, we use TOPCOM\cite{Rambau:TOPCOM:2002} (specifically, the CYTools wrapper for this software). Note: TOPCOM is a sophisticated software, written in C++, that has been developed for nearly $20$ years. The on-demand implementation, in contrast, is written in $\sim$weeks in Python and has had relatively little optimization. Despite these serious handicaps, we will see that the on-demand algorithm achieves order of magnitude speedups (and memory reduction) when compared to TOPCOM's high bar.

To compare the two algorithms, we generate all FRSTs/NTFEs of the first $10$ favorable polytopes of KS (lattice N) for each $5\leq h^{1,1}\leq 10$. The total number of triangulations, peak memory usage, and clock time for each $h^{1,1}$ are all recorded. All calculations are done in a Docker Image ($8$GB RAM) on a M1 Pro chip.

As is visible in \cref{tab:benchmarking}, the count of NTFEs and FRSTs both increase exponentially with $h^{1,1}$, but $N_\mathrm{FRST}$ increases exponentially quicker than $N_\mathrm{NTFE}$. In this benchmark, enumeration of all FRSTs can only be performed for $h^{1,1}\leq8$ since it requires $>8$ GB of memory to enumerate all FRSTs of the selected polytopes with $h^{1,1}\geq9$. Contrast this to the on-demand generation, which never requires more than $15$ MB of memory, orders of magnitude less than the mod algorithm.

\begin{table}[H]
\begin{center}
\begin{tabular}{c | c c | c c | c c} 
 $h^{1,1}$ & \shortstack{$N_\mathrm{FRST}$ \\ (TOPCOM)} & \shortstack{$N_\mathrm{NTFE}$ \\ (CYTools)} & \shortstack{$M_\mathrm{FRST}$ [MB] \\ (TOPCOM)} & \shortstack{$M_\mathrm{NTFE}$ [MB] \\ (CYTools)} & \shortstack{$T_\mathrm{FRST}$ [s] \\ (TOPCOM)} & \shortstack{$T_\mathrm{NTFE}$ [s] \\ (CYTools)}\\
 \hline
 5 & $201$ & $14$ & $1.8$ & $<0.01$ & $1.7$ & $2.4$\\
 6 & $1121$ & $29$ & $94.6$ & $<0.01$ & $19.8$ & $3.2$\\
 7 & $5111$ & $53$ & $775.9$ & $<0.01$ & $118.5$ & $3.3$\\
 8 & $28402$ & $111$ & $5562.9$ & $<0.01$ & $810.8$ & $4.5$\\
 9 & \leaders\hbox{\rule[0.4em]{.1pt}{0.4pt}}\hskip 2em\mbox{}\tablefootnote{Using more memory, one finds that $N_\mathrm{FRST}=256131$ for the first $10$ favorable polytopes at $h^{1,1}=9$.}
 & $570$ & OOM & $0.06$ & \leaders\hbox{\rule[0.4em]{.1pt}{0.4pt}}\hskip 3em\mbox{} & $12.0$\\
 10 & \leaders\hbox{\rule[0.4em]{.1pt}{0.4pt}}\hskip 3em\mbox{} & $5349$ & OOM & $14.7$ & \leaders\hbox{\rule[0.4em]{.1pt}{0.4pt}}\hskip 3em\mbox{} & $95.1$\\
\end{tabular}
\caption{The number of FRSTs/NTFEs of the first $10$ favorable polytopes (in KS ordering) for each $5\leq h^{1,1}\leq 10$, along with the peak memory usage and the clock times needed to calculate them. All calculations are performed in a Docker Image ($8$GB RAM) on a M1 Pro chip. \emph{Counts}: the counts $N_\mathrm{FRST}$ and $N_\mathrm{NTFE}$ both increase exponentially, but $N_\mathrm{FRST}$ increases exponentially quicker. \emph{Memory}: the peak memory appears to increase exponentially (as expected due to the exponentially increasing counts). Even at $h^{1,1}=8$, the memory load to generate all FRSTs is large ($>5$GB); larger values of $h^{1,1}$ lead to crashes due to lack of memory (i.e., out-of-memory or `OOM'). The memory needed to generate all NTFEs is negligible for all studied $h^{1,1}$. \emph{Time}: while TOPCOM is quicker than the on-demand algorithm at $h^{1,1}=5$, this is likely just due to optimization/language differences. At every other $h^{1,1}$, the on-demand algorithm is quicker, with an exponentially increasing speedup as $h^{1,1}$ increases.}
\label{tab:benchmarking}
\end{center}
\end{table}

This demonstrates the first, crucial, point: the on-demand algorithm requires much less memory and thus can probe KS much deeper. For example, with the on-demand algorithm, one can study all NTFE FRSTs for select\footnote{In fact, one can actually push up to $h^{1,1}\approx 30$ but, at such $h^{1,1}$, it becomes increasingly unlikely that a randomly chosen polytope will have few enough NTFE FRSTs.} polytopes with $h^{1,1}=20$, such as is plotted in \cref{fig:big_flip_graph}. Even though we did some extra work here (in calculating the adjacencies/edges), it still only took $<2$s to generate this graph. While this plot displays a small number ($81$) of NTFEs, the associated number of FRSTs is too large to reasonably calculate on one's laptop - the mod algorithm has little hope of getting close to such $h^{1,1}$-scales.

\begin{figure}[t]
    \centering
    \includegraphics[width=0.75\textwidth]{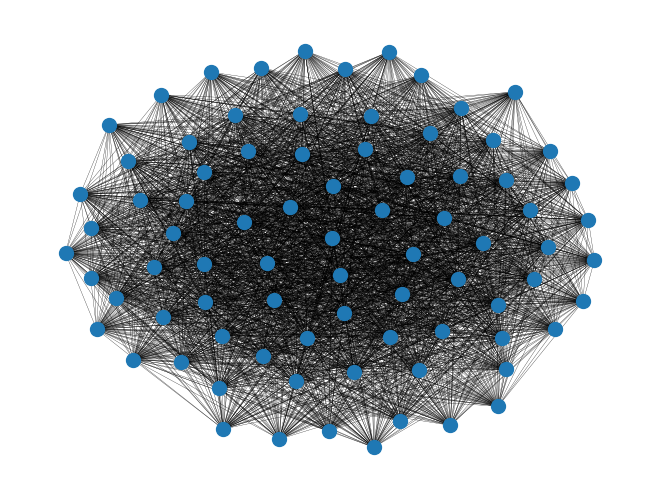}
    \caption{The NTFE flip graph for the $0$th polytope with $h^{1,1}=20$. It is unlikely to calculate this graph (in reasonable time) using the methods of \cref{sec:naive}; one needs \cref{subsec:on-demand}.}
    \label{fig:big_flip_graph}
\end{figure}

The second, crucial, point is the speed: as is visible in \cref{tab:benchmarking}, the on-demand algorithm is exponentially quicker than the mod-algorithm. For example, consider the largest $h^{1,1}$ for which a comparison can be made: $h^{1,1}=8$. Here, it only takes $\sim 4.5$s to generate all $117$ NTFEs of the first $10$ polytopes, compared to the $>13.5$min to generate the $28402$ FRSTs. This speed of the on-demand algorithm (in addition to its minimal memory draw) enables NTFE generation to become somewhat negligible. One can focus on the physics questions at hand, rather than this technical step of generating the NTFEs.

There are three notes of primary importance regarding the computational efficiency of the on-demand algorithm:
\begin{enumerate}
    \item Currently, the on-demand algorithm is limited by the final step, using LP to find points in cones. If one just cares about generating the relevant (intersected) cones, then one can push studies much deeper - without further optimization, the on-demand algorithm generates $\sim700000$ relevant cones per second (no checks for solid-ness in this number). This takes very little memory.
    \item If not all NTFEs are needed, then one can either generate random $2$-face triangulations to `extend' to FR(S)Ts or generate all relevant cones (as above) and sample from said cones. Both methods are significantly more efficient than the corresponding methods of either the brute-force or mod algorithms. The first option (generating random $2$-face triangulationss) is preferred - timing/memory in this case is only limited by the number of points on $2$-faces. Thus, the algorithm scales even to the largest polytope, with $h^{1,1}=491$.
    \item There is significant structure in the cone's hyperplanes that is simply being ignored by our use of an LP solver. E.g., all hyperplanes are sparse (at most $4$ nonzero elements), sum to $0$, and come in one of three patterns (either consisting of elements $\{-1,-1,1,1\}$ or $\{-2N-2,1,1,2N\}$ or $\{-2N-3,1,1,2N+1\}$ for $N\in\mathbb{Z}_{\geq}$; see \cref{app:gerald_verification}).
\end{enumerate}

Clearly, then, despite the on-demand algorithm having numerous handicaps, one finds that it is much quicker than the mod algorithm. Further, we see that, if one can reformulate questions in terms of the relevant cones, then nearly all computational difficulties regarding generation of NTFEs become extremely negligible.

\section{Secondary subfan}
\label{sec:fan}

We have seen that secondary cones define the heights generating regular triangulations. Following this, we studied the set of all secondary cones that could be associated with a polytope, mod $2$-face equivalence. This collection is formally a fan, specifically the fan of all FR(S)Ts. It is conventional to view this as a subfan of only the fine triangulations of the polytope, and the `full' fan as being all (regular) triangulations of the polytope.

This object, in general, is quite complicated: for a polytope $\Delta\in\mathrm{KS}$, the associated secondary subfan is composed of $\#$ NTFE FR(S)Ts cones. For example, each polytope with $h^{1,1}=462$ has $\lesssim 10^{401}$ associated NTFE FR(S)Ts\cite{Demirtas:2020dbm}. Resultantly, their associated subfans would each be made of $\lesssim 10^{401}$ cones. Even if these bounds are exponentially loose, the subfans are still comprised of an infeasible-to-count number of cones.

While such a complete description for generic polytopes is foreseeably impossible (it is equivalent to enumerating all NTFE FR(S)Ts), it actually is not difficult to describe the union of all such cones. This union is itself a polyhedral cone and it has the interpretation that points in it generate fine, regular subdivisions of the polytope. In fact, except for a measure-$0$ set (the boundaries between secondary cones), these height vectors generate FR(S)Ts. Denote (the defining hyperplanes of) this union as $G$.

Our strategy for determining $G$ will be analogous to the on-demand FRST algorithm in \cref{subsec:on-demand}:
\begin{enumerate}
    \item for each $2$-face, $f_i$, of the polytope of interest, $\Delta$,
    \item generate the associated cone defining the height-space of all fine triangulations, denoted $G_i$, and then
    \item calculate $G = \begin{bmatrix} G_1 \Pi_1\\ \vdots\\ G_n \Pi_n \end{bmatrix}$.
\end{enumerate}
This algorithm works for the same reason that the on-demand FRST algorithm works: a regular triangulation of $\Delta$ is fine on all of its $2$-faces if and only if it is fine on each $2$-face. A regular subdivison of a $2$-face is fine if and only if the associated height vector is in $G_i$. Such a technique allows us to restrict attention to $2$D polytopes, for which the algorithm follows simply (\cref{alg:gerald}; verified in \cref{app:gerald_verification}).

\begin{algorithm}[t]
\caption{(support of the) secondary subfan of fine triangulations}
\label{alg:gerald}
\begin{algorithmic}
\State let \textit{H} be an empty array
\State \# (linear constraints)
\For{$\{p_i, p_j\}\subset\mathcal{A}$ for which $\exists!p_r\in\mathcal{A}$ such that $p_r\in\mathrm{int}\left(\mathrm{conv}(\{p_i, p_j\})\right)$}
    \State let $n$ be a vector such that $n^l = \begin{cases}
        1 & l\in\{i,j\}\\
        -2 & l=r\\
        0 & \mathrm{otherwise}.
    \end{cases}$
    \State append $n$ to $H$
\EndFor
\State \# (planar constraints)
\For{$\{p_i, p_j, p_k\}\subset\mathcal{A}$ for which $\exists!p_r\in\mathcal{A}$ such that $p_r\in\mathrm{int}\left(\mathrm{conv}(\{p_i, p_j, p_k\})\right)$}
    \State let $n$ be a vector such that $n^l = \begin{cases}
        1 & l\in\{i,j,k\}\\
        -3 & l=r\\
        0 & \mathrm{otherwise}.
    \end{cases}$
    \State append $n$ to $H$
\EndFor
\State \textbf{return} $H$
\end{algorithmic}
\end{algorithm}

Thus, by simply evaluating \cref{alg:gerald}, one can generate the support of the subfan of fine triangulations of a given polytope. The only difficulty in this algorithm is finding the subsets $S\subset\mathcal{A}$ for which the checks (i.e., exactly one point in the strict interior of $S$; no non-vertices on boundary) apply. This comes down to evaluating the area of triangles and checking if there are lattice points along a line segment. These operations may efficiently be performed using the shoelace formula (+ Pick's theorem) and checks on GCDs.

While this subfan is useful for sampling triangulations, as shown in \cite{Demirtas:2020dbm}, we do not make further comments in this paper on its applications (other than that other algorithms using it are currently being developed).

\FloatBarrier

\section{Conclusion}
\label{sec:conclusion}

The Kreuzer-Skarke database is the largest known collection of Calabi-Yau threefold toric hypersurfaces and our best hope for finding/constructing solutions of string theory with physics similar to our world. Unfortunately, there seems to currently be no clear roadmap\footnote{Correlations have been observed but, as far as the author knows, there are no definitive answers to questions like ``where in KS should one search to find a dS vacuum?"} of where in KS to search for CYs with `interesting physics'. Thus, this wealth of solutions actually raises a very serious computational obstacle: given that we likely can never study every CY in KS, how might one find a CY similar to the real world?

This paper focused on the most direct response: colloquially, ``just perform the calculations quicker by not stepping through an exponentially redundant set". Specifically, following Wall's theorem, this paper demonstrated a simple algorithm (see \cref{subsec:on-demand}) for either directly generating NTFE FR(S)Ts, or proving that no such FR(S)T exists. Despite not being optimized, Python implementations of this algorithm (in CYTools \cite{Demirtas2022cytools}) were observed to be orders of magnitude faster than the algorithm laid out in \cref{sec:naive}, matching expectations (see \cref{fig:redundancy}).

Following this, we demonstrated how to directly generate (the support of) the secondary subfan of fine triangulations in \cref{sec:fan}. Initial tests (not included) suggest that this subfan can be used to sample orders-of-magnitude more NTFE FRSTs of the polytope with $h^{1,1}=491$ than previously achievable, but we do not present detailed discussion here. Instead, \cref{alg:gerald} is provided since this subfan defines the `landscape' of secondary cones, and since it is extremely similar to \cref{alg:cpl_ineqs}.

There are many future directions for this study, some of which are direct:
\begin{enumerate}
    \item better utilize the structure of the hyperplanes when finding points in the relevant cones,
    \item better understand when a set of $2$-face triangulations is `extendable' into an FR(S)T, and
    \item further optimize the implementation of the main algorithm (\cref{sec:on-demand}).
\end{enumerate}
There are also larger scale directions, such as:
\begin{enumerate}
    \item reformulate string-theory questions into those relating the (intersections of) secondary cones,
    \item further develop algorithms exploiting the secondary subfan of fine triangulations,
    \item look (especially in $2$-face data) for structures/correlations relevant to physical quantities, and
    \item look for similar algorithms that cut out even more redundancies.
\end{enumerate}
As one final note, despite KS presenting daunting computational challenges, we see that there is much room to grow.

\section{Acknowledgements}

I would like to thank Liam McAllister and Naomi Gendler for their discussions that ultimately led me down this line of study, and their encouragement throughout it. Further, I would like to thank Andres Rios-Tascon for his greatly enjoyable (and informative) conversations about secondary cones and \cref{alg:cpl_ineqs}. I would also like to thank Andreas Schachner for his encouragement and willingness to always discuss such computational topics. Finally, I would like to thank Jakob Moritz for his discussions about regular triangulations as represented in height space.

This work was funded in part by NSF grant PHY–2014071.

\bibliographystyle{utphys}
\bibliography{refs}

@article{Wall:1966,
    Author = {Wall, C.  T.  C. },
    Da = {1966/12/01},
    Doi = {10.1007/BF01389738},
    Id = {Wall1966},
    Isbn = {1432-1297},
    Journal = {Inventiones mathematicae},
    Number = {4},
    Pages = {355--374},
    Title = {Classification problems in differential topology. V},
    Ty = {JOUR},
    Url = {https://doi.org/10.1007/BF01389738},
    Volume = {1},
    Year = {1966}
}

@book{Rockafellar:1970,
    url = {https://doi.org/10.1515/9781400873173},
    title = {Convex Analysis},
    author = {Ralph Tyrell Rockafellar},
    publisher = {Princeton University Press},
    address = {Princeton},
    doi = {doi:10.1515/9781400873173},
    isbn = {9781400873173},
    year = {1970},
    lastchecked = {2023-09-10}
}

@article{Reznick:1986,
    title = {Lattice point simplices},
    journal = {Discrete Mathematics},
    volume = {60},
    pages = {219-242},
    year = {1986},
    issn = {0012-365X},
    doi = {https://doi.org/10.1016/0012-365X(86)90015-4},
    url = {https://www.sciencedirect.com/science/article/pii/0012365X86900154},
    author = {Bruce Reznick}
}

@article{Berglund:1995,
    title = {Mirror symmetry and the moduli space for generic hypersurfaces in toric varieties},
    journal = {Nuclear Physics B},
    volume = {456},
    number = {1},
    pages = {153-204},
    year = {1995},
    issn = {0550-3213},
    doi = {https://doi.org/10.1016/0550-3213(95)00434-2},
    url = {https://www.sciencedirect.com/science/article/pii/0550321395004342},
    author = {Per Berglund and Sheldon Katz and Albrecht Klemm}
}

@article{Kreuzer:2000xy,
    author = "Kreuzer, Maximilian and Skarke, Harald",
    title = "{Complete classification of reflexive polyhedra in four-dimensions}",
    eprint = "hep-th/0002240",
    archivePrefix = "arXiv",
    reportNumber = "HUB-EP-00-13, TUW-00-07",
    doi = "10.4310/ATMP.2000.v4.n6.a2",
    journal = "Adv. Theor. Math. Phys.",
    volume = "4",
    pages = "1209--1230",
    year = "2000"
}

@InProceedings{Rambau:TOPCOM:2002,
  Title                    = {TOPCOM: Triangulations of Point Configurations and Oriented Matroids},
  Author                   = {Rambau, J{\"o}rg},
  Booktitle                = {Proceedings of the International Congress of Mathematical Software},
  Year                     = {2002},
  URL                      = {http://www.zib.de/PaperWeb/abstracts/ZR-02-17}
}

@book{Loera:2010,
    title = "Triangulations",
    subtitle = "Structures for Algorithms and Applications",
    author = "Loera, Jes{\'{u}}s A. and Rambau, J{\"{o}}rg and Santos, Francisco",
    year = "2010",
    publisher = "Springer Berlin, Heidelberg",
    doi = {10.1007/978-3-642-12971-1},
    isbn = {978-3-642-12971-1}
}

@misc{Braun:2017,
      title={The Hodge Numbers of Divisors of Calabi-Yau Threefold Hypersurfaces}, 
      author={Andreas P. Braun and Cody Long and Liam McAllister and Michael Stillman and Benjamin Sung},
      year={2017},
      eprint={1712.04946},
      archivePrefix={arXiv},
      primaryClass={hep-th}
}

@article{Demirtas:2020dbm,
    author = "Demirtas, Mehmet and McAllister, Liam and Rios-Tascon, Andres",
    title = "{Bounding the Kreuzer-Skarke Landscape}",
    eprint = "2008.01730",
    archivePrefix = "arXiv",
    primaryClass = "hep-th",
    doi = "10.1002/prop.202000086",
    journal = "Fortsch. Phys.",
    volume = "68",
    pages = "2000086",
    year = "2020"
}

@article{Demirtas:2021gsq,
    author = "Demirtas, Mehmet and Gendler, Naomi and Long, Cody and McAllister, Liam and Moritz, Jakob",
    title = "{PQ axiverse}",
    eprint = "2112.04503",
    archivePrefix = "arXiv",
    primaryClass = "hep-th",
    doi = "10.1007/JHEP06(2023)092",
    journal = "JHEP",
    volume = "06",
    pages = "092",
    year = "2023"
}

@misc{Demirtas2022cytools,
      title={CYTools: A Software Package for Analyzing Calabi-Yau Manifolds}, 
      author={Mehmet Demirtas and Andres Rios-Tascon and Liam McAllister},
      year={2022},
      eprint={2211.03823},
      archivePrefix={arXiv},
      primaryClass={hep-th}
}

@article{Gendler:2022ztv,
    author = "Gendler, Naomi and Heidenreich, Ben and McAllister, Liam and Moritz, Jakob and Rudelius, Tom",
    title = "{Moduli Space Reconstruction and Weak Gravity}",
    eprint = "2212.10573",
    archivePrefix = "arXiv",
    primaryClass = "hep-th",
    reportNumber = "ACFI-T22-10",
    month = "12",
    year = "2022"
}

@article{Gendler:2023hwg,
    author = "Gendler, Naomi and Janssen, Oliver and Kleban, Matthew and La Madrid, Joan and Mehta, Viraf M.",
    title = "{Axion minima in string theory}",
    eprint = "2309.01831",
    archivePrefix = "arXiv",
    primaryClass = "hep-th",
    month = "9",
    year = "2023"
}

@article{Gendler:WIP,
    author = "Gendler, Naomi and MacFadden, Nate and McAllister, Liam and Moritz, Jakob and Nally, Richard and Schachner, Andreas and Stillman, Mike.",
    title = "{Counting Calabi-Yau Threefolds}",
    note = "WIP"
}

@article{MacFadden:WIP,
    author = "MacFadden, Nate and Stepniczka, Michael",
    title = "{Bounding Calabi-Yau Threefolds at $h^{1,1}=491$}",
    note = "WIP"
}

\newpage
\appendix
\section{Verification of \cref{alg:gerald}}
\label{app:gerald_verification}

In this section, we show that the output of \cref{alg:gerald}, denoted $H$, defines the support of the secondary subfan of fine triangulations. As argued in \cref{sec:fan}, attention can be limited to $2$D cases - we will do so throughout this section\footnote{So the $G$ here is akin to $G_i$ in the main text.}. That is, we will show $\mathrm{coni}(G)^* = \mathrm{coni}(H)^*$, where the conical hulls are taken over the rows of each matrix.

\subsection{$\mathrm{coni}(G)^*\subseteq\mathrm{coni}(H)^*$}
Here, we show that the support of the secondary subfan of fine triangulations is contained in the cone defined by $H$.

\begin{proposition}
    $\mathrm{coni}(G)^*\subseteq\mathrm{coni}(H)^*$.
\end{proposition}

The proof of this proposition effectively is just a recognition of what the hyperplanes in $H$ mean. Recall that there are two types of hyperplanes in $H$ (see \cref{alg:gerald,fig:gerald_constraints}):
\begin{enumerate}
    \item `linear constraint' $\to$ for every line segment $\{p_i,p_j\}$ containing exactly $1$ interior point, $p_r$, ensure that $\tilde{p}_r^3\leq\tilde{q}^3$ for the point $\tilde{q}\in\mathrm{conv}(\{\tilde{p}_i,\tilde{p}_j\})$ with $\tilde{q}^1=\tilde{p}_r^1$ and $\tilde{q}^2=\tilde{p}_r^2$, and
    \item `plane constraint' $\to$ for every triangle $\{p_i,p_j,p_k\}$ containing exactly $1$ point, $p_r$, in its strict interior (and no other lattice points on the boundary), ensure that $\tilde{p}_r^3\leq\tilde{q}^3$ for the point $\tilde{q}\in\mathrm{conv}(\{\tilde{p}_i,\tilde{p}_j,\tilde{p}_k\})$ with $\tilde{q}^1=\tilde{p}_r^1$ and $\tilde{q}^2=\tilde{p}_r^2$.
\end{enumerate}
The latter constraint warrants a brief note: in \cref{alg:gerald}, we implicitly use that $p_r = (p_i + p_j + p_k)/3$. That is, that $p_r$ is the centroid of the triangle $\{p_i,p_j,p_k\}$. This is always true for $2$D lattice triangles whose only non-vertex lattice point is in the strict interior\cite{Reznick:1986}.

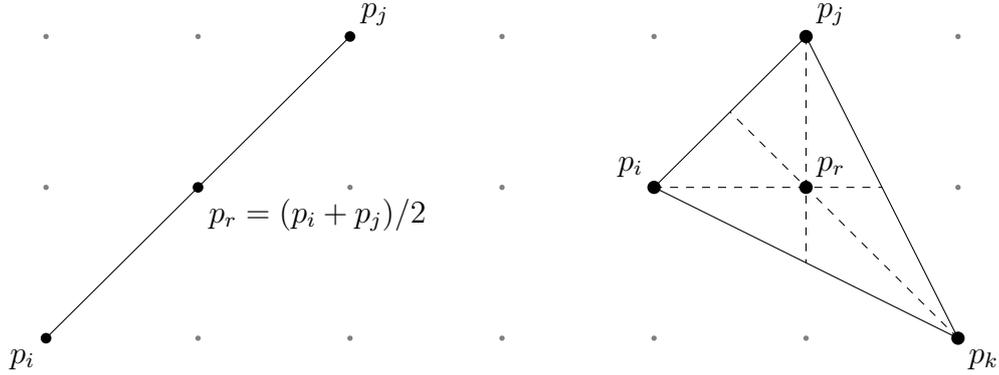
\begin{figure}[h]
    \centering
    \begin{tikzpicture}
        \foreach \x in {0,2,4,6,8,10,12}
            \foreach \y in {0,2,4}
                \fill[gray] (\x,\y) circle (1pt);
    
        \coordinate (A) at (0,0);
        \coordinate (B) at (4,4);
        \coordinate (M) at (2,2);
    
        \draw (A) -- (B);
        \foreach \p in {A, M, B}
            \fill (\p) circle(2pt);
    
        \node[below left] at (A) {$p_i$};
        \node[above right] at (B) {$p_j$};
        \node[below right] at (M) {$p_r = (p_i + p_j)/2$};
    
        \coordinate (Pi) at (8,2);
        \coordinate (Pj) at (10,4);
        \coordinate (Pk) at (12,0);
        \coordinate (Pr) at (10,2);
    
        \coordinate (Midij) at ($(Pi)!0.5!(Pj)$);
        \coordinate (Midik) at ($(Pi)!0.5!(Pk)$);
        \coordinate (Midjk) at ($(Pj)!0.5!(Pk)$);
    
        \draw (Pi) -- (Pj) -- (Pk) -- cycle;
        \draw[dashed] (Pi) -- (Midjk);
        \draw[dashed] (Pj) -- (Midik);
        \draw[dashed] (Pk) -- (Midij);
    
        \foreach \p in {Pi, Pj, Pk, Pr}
            \fill (\p) circle(2.5pt);
    
        \node[above left] at (Pi) {$p_i$};
        \node[above right] at (Pj) {$p_j$};
        \node[below right] at (Pk) {$p_k$};
        \node[above right] at (Pr) {$p_r$};
        
    \end{tikzpicture}
    \caption{Illustrations of the linear and planar constraints in \cref{alg:gerald}.}
    \label{fig:gerald_constraints}
\end{figure}

\begin{proof}
    Let $h\in\mathrm{coni}(G)^*$. That is, $h$ is a height vector for $\mathcal{A}$ such that all points $\tilde{p}\in\tilde{\mathcal{A}}$ appear on some lower face of $\mathrm{conv}(\tilde{\mathcal{A}})$.
    
    Assume that $h\notin\mathrm{coni}(H)^*$. Thus, a line constraint and/or a plane constraint is violated. That means, there is a lifted lattice point, $\tilde{p}_r$, and another point $\tilde{q}$ such that $\tilde{q}^1=\tilde{p}_r^1$ and $\tilde{q}^2=\tilde{p}_r^2$ but $\tilde{q}^3<\tilde{p}_r^3$. Thus, $\tilde{p}_r$ cannot be on a lower face, a contradiction.
\end{proof}

\subsection{$\mathrm{coni}(G)\subseteq\mathrm{coni}(H)$}

Here, we show that the cone defined by $H$ is contained in the support of the secondary subfan of fine triangulations. We do so through use of the dual, instead showing that $\mathrm{coni}(G)\subseteq\mathrm{coni}(H)$. That is, we want to show that every hyperplane inequality of $\mathrm{coni}(G)^*$ can be represented as the conical hull of hyperplane inequalities in $\mathrm{coni}(H)^*$. First, we must understand what the hyperplane inequalities of $\mathrm{coni}(G)^*$ look like.

\begin{lemma}
\label{lem:secondary_hyperplanes_are_ok}
    $\mathrm{coni}(G)\subseteq\sum_{i=1}^N \mathrm{coni}(H_{\mathcal{T}_i})$, where $H_{\mathcal{T}_i}$ are the hyperplane inequalities associated to the secondary cone of the $i$th FRST of the associated polytope.
\end{lemma}

\begin{proof}
    Recall, $G$ describes the union or, (in this case) equivalently, the sum of all secondary cones:
    \begin{equation}
        \mathrm{coni}(G)^* = \bigcup_{i=1}^N \mathrm{coni}(H_{\mathcal{T}_i})^* = \sum_{i=1}^N \mathrm{coni}(H_{\mathcal{T}_i})^*.
    \end{equation}
    It is not hard to see\cite[Corollary 16.4.2]{Rockafellar:1970} that this is, equivalently
    \begin{equation}
        \mathrm{coni}(G) = \left(\sum_{i=1}^N \mathrm{coni}(H_{\mathcal{T}_i})^*\right)^* = \bigcap_{i=1}^N \mathrm{coni}(H_{\mathcal{T}_i}).
    \end{equation}
    Thus
    \begin{equation}
        \mathrm{coni}(G) = \bigcap_{i=1}^N \mathrm{coni}(H_{\mathcal{T}_i}) \subseteq \sum_{i=1}^N \mathrm{coni}(H_{\mathcal{T}_i}).
    \end{equation}
\end{proof}

We use \cref{lem:secondary_hyperplanes_are_ok} to reformulate our studies: instead of directly studying hyperplanes in $\mathrm{coni}(G)^*$, we can instead study those in $\mathrm{coni}(H_{\mathcal{T}_i})^*$. We do so because there is a nice characterization (\cref{alg:cpl_ineqs}) of the inequalities in $\mathrm{coni}(H_{\mathcal{T}_i})^*$. Explicitly, $\mathrm{coni}(H_{\mathcal{T}_i})$ is generated by rays, each ray associated to a pair, $T_i$ and $T_j$, of unimodular simplices sharing an edge (see \cref{fig:cpl_cases}). There are three brief notes warranting mention regarding the inequalities associated to such $T_i$ and $T_j$:
\begin{enumerate}
    \item for a given triangulation $\mathcal{T}$, there is a single inequality in $\mathrm{coni}(H_{\mathcal{T}})^*$ associated to $T_i$ and $T_j$. This is because the rank of the matrix in \cref{alg:cpl_ineqs} is $3$,
    \item the inequalities in \cref{alg:cpl_ineqs} are basis independent so, when calculating them, we can always pick a lattice basis such that $T_i\cap T_j$ is horizontal and such that $T_i$ also has a vertical edge (exactly as in \cref{fig:cpl_cases}), and
    \item for $T_j$ to be unimodular $p_{n2}^2 = p_{s1}^2-1 = p_{s2}^2-1$. That is, $T_j$ has a `height' of $1$.
\end{enumerate}
From these notes, a simple lemma follows (see \cref{fig:pts_are_horizontal}):

\begin{figure}[t]
    \centering
    \begin{tikzpicture}
        \foreach \x in {0,2,4,6,8,10,12}
            \foreach \y in {0,2,4}
                \fill[gray] (\x,\y) circle (1pt);
    
        \draw (2,4) -- (4,2) -- (2,2) -- cycle;
        \draw (0,0) -- (4,2) -- (2,2) -- cycle;
        \draw[dashed,gray] (0,0) -- (2,4);

        \node at (2.6,2.6) {$T_i$};
        \node at (2.1,1.5) {$T_j$};

        \fill (2,4) circle(2.5pt);
        \fill (2,2) circle(2.5pt);
        \fill (4,2) circle(2.5pt);
        \fill (0,0) circle(2.5pt);
    
        \node[left] at (2,4) {$p_{n1}$};
        \node[above left] at (2,2) {$p_{s1}$};
        \node[above right] at (4,2) {$p_{s2}$};
        \node[left] at (0,0) {$p_{n2}$};
        
        \draw (6,4) -- (8,2) -- (6,2) -- cycle;
        \draw (6,0) -- (8,2) -- (6,2) -- cycle;
        \draw[dashed,gray,line width=3pt] (6,0) -- (6,4);

        \node at (6.6,2.6) {$T_i$};
        \node at (6.6,1.4) {$T_j$};

        \fill (6,4) circle(2.5pt);
        \fill (6,2) circle(2.5pt);
        \fill (8,2) circle(2.5pt);
        \fill (6,0) circle(2.5pt);
    
        \node[left] at (6,4) {$p_{n1}$};
        \node[left] at (6,2) {$p_{s1}$};
        \node[above right] at (8,2) {$p_{s2}$};
        \node[left] at (6,0) {$p_{n2}$};

        \draw (10,4) -- (12,2) -- (10,2) -- cycle;
        \draw (12,0) -- (12,2) -- (10,2) -- cycle;
        \draw[dashed,gray] (10,4) -- (12,0);

        \node at (11,2.5) {$T_i$};
        \node at (10.9,1.5) {$T_j$};

        \fill (10,4) circle(2.5pt);
        \fill (10,2) circle(2.5pt);
        \fill (12,2) circle(2.5pt);
        \fill (12,0) circle(2.5pt);
    
        \node[left] at (10,4) {$p_{n1}$};
        \node[left] at (10,2) {$p_{s1}$};
        \node[above right] at (12,2) {$p_{s2}$};
        \node[left] at (12,0) {$p_{n2}$};

        \node at (-1,2) {{\huge $\dots$}};
        \node at (13,2) {{\huge $\dots$}};
        
    \end{tikzpicture}
    \caption{The geometries of the simplices $T_i$ and $T_j$ defining the secondary cone inequalities. These shapes are classified by where the dotted line, $\{p_{n1},p_{n2}\}$ lays. If it lays outside the convex hull (left), then the case is called `two triangles'; if it lays on an edge of the convex hull (center), then the case is called `one triangle'; otherwise (right), the case is called `parallelogram'.}
    \label{fig:cpl_cases}
\end{figure}
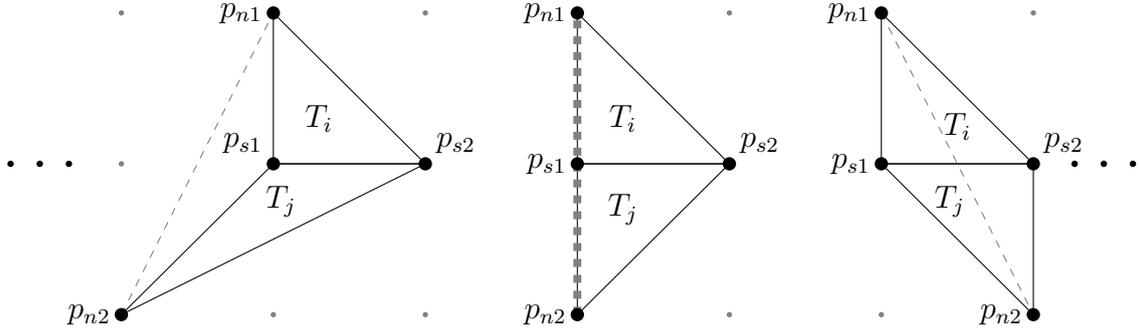

\begin{lemma}
\label{lem:pts_are_horizontal}
    Any non-vertex lattice point $p\in\mathrm{conv}(T_i\cup T_j)$ has $y$-coordinate equal to $p_{s1}^2$ (equivalently, $p_{s2}^2$).
\end{lemma}

\begin{proof}
    Take the aforementioned choice of basis ($\{p_{s1},p_{s2}\}$ is horizontal, $\{p_{s1},p_{n1}\}$ is vertical). Let $y_0 = p_{s1}^2 = p_{s2}^2$.

    Then, as previously noted, all points $(x,y)\in\mathrm{conv}(T_i\cup T_j)$ obey $y_0-1\leq y\leq y_0+1$. Furthermore, the only point with $y$-coordinate equal to $y_0-1$ is $p_{n2}$; the only point with $y$-coordinate equal to $y_0+1$ is $p_{n1}$.
\end{proof}

\begin{figure}[t]
    \centering
    \begin{tikzpicture}
        \foreach \x in {0,...,12}
            \foreach \y in {0,2,4}
                \fill[gray] (\x,\y) circle (1pt);
    
        \draw (5,4) -- (6,2) -- (5,2) -- cycle;
        \draw (0,0) -- (6,2) -- (5,2) -- cycle;
        \draw[dashed,gray] (0,0) -- (5,4);

        \fill (5,4) circle(2.5pt);
        \fill (5,2) circle(2.5pt);
        \fill (6,2) circle(2.5pt);
        \fill (0,0) circle(2.5pt);

        \node[draw=red, cross=2.5pt] at (3,2) {};
        \node[draw=red, cross=2.5pt] at (4,2) {};
    
        \node[left] at (5,4) {$p_{n1}$};
        \node[above left] at (5,2) {$p_{s1}$};
        \node[above right] at (6,2) {$p_{s2}$};
        \node[left] at (0,0) {$p_{n2}$};

        \node[red, below] at (3,2) {$q_1$};
        \node[red, below left] at (4,2) {$q_2$};

        \draw (11,4) -- (12,2) -- (11,2) -- cycle;
        \draw (7,0) -- (12,2) -- (11,2) -- cycle;
        \draw[dashed,gray] (7,0) -- (11,4);

        \fill (11,4) circle(2.5pt);
        \fill (11,2) circle(2.5pt);
        \fill (12,2) circle(2.5pt);
        \fill (7,0) circle(2.5pt);

        \node[draw=red, cross=2.5pt] at (9,2) {};
        \node[draw=red, cross=2.5pt] at (10,2) {};
    
        \node[left] at (11,4) {$p_{n1}$};
        \node[above left] at (11,2) {$p_{s1}$};
        \node[above right] at (12,2) {$p_{s2}$};
        \node[left] at (7,0) {$p_{n2}$};

        \node[red, below] at (9,2) {$q_1$};
        \node[red, below left] at (10,2) {$q_2$};
        
    \end{tikzpicture}
    \caption{For two unimodular triangle $T_i=\{p_{n1},p_{s1},p_{s2}\}$ and $T_j=\{p_{n2},p_{s1},p_{s2}\}$ in a lattice basis such that $p_{s1}^2=p_{s2}^2$ and $p_{s1}^1=p_{n1}^1$, any non-vertex lattice point, $p\in\mathrm{conv}(T_i\cup T_j)$ must have $p^2 = p_{s1}^2$. I.e., it must lay on the (extended) line $T_i\cap T_j$. There are then two cases: there can be a lattice point on the interior of the line $\{p_{n1},p_{n2}\}$ (right) or not (left).}
    \label{fig:pts_are_horizontal}
\end{figure}
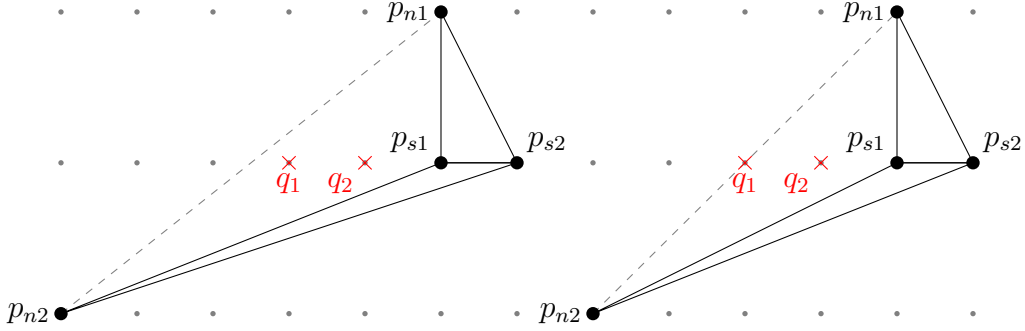

It will be beneficial to classify such inequalities (i.e., pairs of simplices $T_i$ and $T_j$) by their shape and study them case-by-case. Recognize that such a pair of simplices is exactly one of the following three classes:
\begin{enumerate}
    \item (right of \cref{fig:cpl_cases}; `parallelogram'): $\mathrm{conv}(T_i\cup T_j) = T_i\cup T_j$ and both diagonals are strictly contained in the region $T_i\cup T_j$,
    \item (center of \cref{fig:cpl_cases}; `one-triangle'): $\mathrm{conv}(T_i\cup T_j) = T_i\cup T_j$ and one `diagonal' is also an edge of the region $T_i\cup T_j$, or
    \item (left of \cref{fig:cpl_cases}; a `two-triangles'): $\mathrm{conv}(T_i\cup T_j)\supset T_i\cup T_j$.
\end{enumerate}
As we will immediately see, inequalities associated with the parallelogram case do not arise in $\mathrm{coni}(G)^*$, while the other two do. It is not hard to see, however, that the other two classes of inequalities are in $\mathrm{coni}(H)$. Since these classes of inequalities generate $\mathrm{coni}(G)$, this then will show $\mathrm{coni}(G)\subseteq\mathrm{coni}(H)$.

\subsubsection{Parallelogram case}

Begin with the parallelogram case (right side of \cref{fig:cpl_cases}). There are no concerns about fine-ness of points - each point is a vertex. The associated inequality just controls which diagonal occurs (the solid line vs. the dashed line). Thus, since we are only concerned with fineness, we expect that such inequalities do not occur in $\mathrm{coni}(G)$.

To show this, we split analysis into whether the diagonal is flippable (i.e., into another regular triangulation) or not. First, consider the case that it is flippable, then the secondary cones associated to these two regular triangulations share a facet (a `wall'). The inequality associated with the parallelogram is normal to this wall and the two cones on either side have opposite normal vectors (but not both). Thus, these normals are not contained in
\begin{equation}
    \mathrm{coni}(G) = \bigcap_{i=1}^N \mathrm{coni}(H_{\mathcal{T}_i}).
\end{equation}

Second, consider the case that the diagonal is not flippable. Let $n_{T_i,T_j}$ be the inequality associated to the pair of triangles $T_i$ and $T_j$. The diagonal not being flippable means that there is no such $h$ such that $(H_\mathcal{T}\setminus n_{T_i,T_j})\,h\geq 0$ and $n_{T_i,T_j}\cdot h<0$. That is, the set of all other inequalities, $H_\mathcal{T}\setminus n_{T_i,T_j}$, implies that $n_{T_i,T_j}\cdot h\geq0$. Thus, explicit consideration of this inequality is not necessary - it will be implied by the other cases (i.e., one and two triangle(s)).

\subsubsection{One-triangle case}
Now, the one-triangle case. We will show that the inequality, $n$, associated to the center of \cref{fig:cpl_cases} is just the linear constraint from \cref{alg:gerald} and \cref{fig:gerald_constraints}.

This follows simply from the fact that the null space is $1$D and that
\begin{equation}
    n^r = 
    \begin{cases}
        1 & r\in\{n1,n2\}\\
        -2 & r=s1\\
        0 & \mathrm{otherwise}
    \end{cases}
\end{equation}
spans this null space. There is a simple explanation as to why this $n$ is just the same vector calculated in \cref{alg:gerald}: the line $\{\tilde{p}_{s1}, \tilde{p}_{s2}\}$ occurs in $\mathrm{conv}(\tilde{T}_i\cup\tilde{T}_j)$ if and only if $\tilde{p}_{s1}$ is on or below the line $\{\tilde{p}_{n1},\tilde{p}_{n2}\}$. Thus, this case is really the same as the linear constraint of \cref{alg:gerald}.

\subsubsection{Two-triangles case}
Now, the two-triangles case. We will show that the inequality, $n$, associated to the right of \cref{fig:cpl_cases} can be written as a conical combination of linear and planar constraints from \cref{alg:gerald}. In contrast to the parallelogram and triangle case, this requires a little work.

Note: if the region $\mathrm{conv}(T_i\cup T_j)$ only contained $4$ lattice points (those labelled in \cref{fig:cpl_cases}), then the associated hyperplane normal is just the planar constraint of \cref{alg:gerald}. Thus, we focus on the case where there are $5+$ points in this region. Note, by \cref{lem:pts_are_horizontal}, all of these points, $p$, must have $p^2 = p_{s1}^2 = p_{s2}^2$. That is, we are considering cases exactly like what is displayed in \cref{fig:pts_are_horizontal}.

Begin with the case where there is a lattice point on the interior of line $\{p_{n1},p_{n2}\}$ (right side of \cref{fig:pts_are_horizontal}). In this case, \cref{alg:gerald} calculates the hyperplanes:
\begin{equation}
    H = \begin{bmatrix}
        n_1\\ n_2\\ \vdots\\ n_{N+1}
    \end{bmatrix} =
    \begin{blockarray}{ccccccccc}
        n1 & n2 & q1 & q2 & q3 & \dots & qN & s1 & s2 \\
        \begin{block}{[ccccccccc]}
          1 & 1 & -2 & 0 & 0  & \cdots & 0 & 0 & 0\\
          0 & 0 & 1 & -2 & 1 & \cdots & 0 & 0 & 0\\
          0 & 0 & 0 & 1 & -2  & \cdots & 0 & 0 & 0\\
          0 & 0 & 0 & 0 & 1 & \cdots & 0 & 0 & 0\\
          \vdots &  &  &  &  & \ddots &  &  & \vdots\\
          0 & 0 & 0 & 0 & 0 & \cdots & 1 & 0 & 0\\
          0 & 0 & 0 & 0 & 0 & \cdots & -2 & 1 & 0\\
          0 & 0 & 0 & 0 & 0 & \cdots & 1 & -2 & 1.\\
        \end{block}
    \end{blockarray}
\end{equation}
Consider the conical combination
\begin{equation}
    m = n_1 + \sum_{k=2}^{N+1} 2(k-1) n_k = \dots = \begin{bmatrix} 1 & 1 & 0 & \cdots & 0 & -(2N+2) & 2N \end{bmatrix}.
\end{equation}
This spans the relevant null space:
\begin{equation}
    p\cdot m = (-2N,-1) + (0,1) - 2(N+1) (0,0) + 2N (1,0) = 0.
\end{equation}
Since the null space is $1$D, we have thus shown that (in this case), the $m\in\mathrm{coni}(H_{\mathcal{T}_i})$ is expressible as a conical combination of $n\in\mathrm{coni}(H)$.

Now, the case where there is not a lattice point on the interior of line $\{p_{n1},p_{n2}\}$ (left side of \cref{fig:pts_are_horizontal}). In this case, \cref{alg:gerald} calculates the hyperplanes:
\begin{equation}
    H = \begin{bmatrix}
        n_1\\ n_2\\ \vdots\\ n_{N+1}
    \end{bmatrix} =
    \begin{blockarray}{ccccccccc}
        n1 & n2 & q1 & q2 & q3 & \dots & qN & s1 & s2 \\
        \begin{block}{[ccccccccc]}
          1 & 1 & -3 & 1 & 0  & \cdots & 0 & 0 & 0\\
          0 & 0 & 1 & -2 & 1 & \cdots & 0 & 0 & 0\\
          0 & 0 & 0 & 1 & -2  & \cdots & 0 & 0 & 0\\
          0 & 0 & 0 & 0 & 1 & \cdots & 0 & 0 & 0\\
          \vdots &  &  &  &  & \ddots &  &  & \vdots\\
          0 & 0 & 0 & 0 & 0 & \cdots & 1 & 0 & 0\\
          0 & 0 & 0 & 0 & 0 & \cdots & -2 & 1 & 0\\
          0 & 0 & 0 & 0 & 0 & \cdots & 1 & -2 & 1.\\
        \end{block}
    \end{blockarray}
\end{equation}
Consider the conical combination
\begin{equation}
    m = n_1 + \sum_{k=2}^{N+1} (2k-1) n_k = \dots = \begin{bmatrix} 1 & 1 & 0 & \cdots & 0 & -(2N+3) & 2N+1 \end{bmatrix}.
\end{equation}
This spans the relevant null space (see \cref{alg:cpl_ineqs}):
\begin{equation}
    p\cdot m = (-2N-1,-1) + (0,1)  - (2N+3) (0,0) + (2N+1) (1,0) = 0.
\end{equation}
Since the null space is $1$D, we have thus shown that (in this case), the $m\in\mathrm{coni}(H_{\mathcal{T}_i})$ is expressible as a conical combination of $n\in\mathrm{coni}(H)$.

Thus, in every $2$-triangle case, the derived secondary cone hyperplane is expressible in terms of constraints in \cref{alg:gerald}.

\subsubsection{All together}
In each case of \cref{fig:cpl_cases}, we saw that the associated hyperplane normal either did not occur in the support of the secondary subfan of fine triangulations, or that it was expressible as the conical hull of hyperplane normals in \cref{alg:gerald}. This thus shows the desired inclusion, that $\mathrm{coni}(G)\subseteq\mathrm{coni}(H)$.

\end{document}